\documentclass[10pt,a4paper]{article}

\usepackage[margin=1in]{geometry}
\usepackage{graphicx}
\usepackage[colorlinks=true,allcolors=black]{hyperref}
\usepackage{amsmath,amsthm,amsfonts,amssymb,amscd,mathrsfs}
\usepackage{enumerate}
\usepackage{enumitem}
\usepackage{bbm}
\usepackage{float}
\usepackage{fancyvrb}
\usepackage{booktabs}
\usepackage{subcaption}
\usepackage{algpseudocode,algorithmicx,algorithm}
\usepackage{etoolbox}
\usepackage{arydshln}
\usepackage{mathtools}
\patchcmd{\abstract}{\small}{}{}{}
\newtheorem{theorem}{Theorem}[section]

\newtheorem{corollary}[theorem]{Corollary}
\newtheorem{proposition}[theorem]{Proposition}
\newtheorem{lemma}[theorem]{Lemma}
\newtheorem{assumption}[theorem]{Assumption}

\newtheorem{remark}[theorem]{Remark}
\newenvironment{breakableproblem}[1]{%
    \refstepcounter{theorem}%
    \par\addvspace{\topsep}%
    \noindent\textit{Problem~\thetheorem\ (#1):}\enspace\ignorespaces
}{%
    \par\addvspace{\topsep}%
}

\DeclareMathOperator{\rank}{rank}
\DeclareMathOperator{\col}{col}

\DeclareMathOperator{\vecs}{vecs}
\DeclareMathOperator{\spanop}{span}
\newcommand{\R}{\mathbb{R}}

\newcommand{\T}{^{\top}}

\allowdisplaybreaks[4]
\title{Reinforcement Learning-Based Output Feedback LQR for Continuous-Time MIMO Systems%
\thanks{This work was supported by Guangdong Basic and Applied Basic Research Foundation (2024A1515010193, 2025B1515120042), National Natural Science Foundation of China (62350710214, 62121004, 62273100, 62550003, U23A20325), Guangdong Science and Technology Program (2024B1212010002), and Shenzhen Science and Technology Program (KQTD20221101093557010).}%
\thanks{This work has been submitted to the IEEE for possible publication. Copyright may be transferred without notice, after which this version may no longer be accessible.}}
\date{\today}
\author{%
Qihua Chen\thanks{Qihua Chen, Mingxiang Liu, Lili Wang, and Zhiyun Lin are with the Guangdong Provincial Key Laboratory of Fully Actuated System Control Theory and Technology, School of Automation and Intelligent Manufacturing, Southern University of Science and Technology, Shenzhen 518055, P. R. China (e-mails: 12332648@mail.sustech.edu.cn; liumx@ieee.org; wangll@sustech.edu.cn; linzy@sustech.edu.cn). Corresponding author: Zhiyun Lin.}
~~~~Mingxiang Liu\footnotemark[3]
~~~~Lili Wang\footnotemark[3]\\
Zhiyun Lin\footnotemark[3]
~~~~and Minyue Fu\thanks{Minyue Fu is with the School of Automation, Guangdong University of Technology, Guangzhou 510006, China (e-mail: brianfu2106@outlook.com).}}

\begin{document}

\maketitle

\begin{abstract}
This article studies model-free output feedback linear quadratic regulation (LQR) for continuous-time linear systems with an $n$-dimensional state, an $m$-dimensional input, and a $p$-dimensional output, using filtered input--output data. Since the system state is unavailable, existing methods rely on dynamic filters to parameterize the hidden state using measurable input--output signals. However, the intrinsic dimension of the resulting filter-based parametrization can be smaller than the dimension of the complete filtered vector, and this deterministic redundancy can make the Bellman regressions rank deficient. We characterize this intrinsic dimension and show that the conventional filtered vector contains only $2n$ independent components for single-input multi-output (SIMO) systems and $n(m+1)$ independent components for general multi-input multi-output (MIMO) systems. Based on this characterization, a reduced filtered vector is extracted directly from data and used to develop reduced model-free output feedback policy iteration and value iteration equations, eliminating the redundant directions and decreasing the number of unknown parameters while retaining a fully input--output data-based implementation. A numerical example illustrates the rank reduction and the effectiveness of the learned controller.

\end{abstract}

\section{Introduction}
Linear quadratic regulation (LQR) is one of the fundamental problems in optimal control. For continuous-time linear systems with known dynamics and measurable states, the optimal state-feedback controller can be obtained by solving an algebraic Riccati equation (ARE)~\cite{Lewis2012}. However, this classical solution requires known system matrices, and its implementation usually relies on full-state feedback. In many practical applications, the system dynamics are unknown and the internal state is not directly measurable. These limitations have motivated the development of data-driven and reinforcement learning (RL)-based methods for optimal control~\cite{SuttonBarto2018,Bertsekas2019}.

Adaptive dynamic programming (ADP) and RL provide an effective framework for solving optimal control problems without explicitly identifying the system model~\cite{Liu2021Survey,Wang2024Survey}. For continuous-time systems, integral reinforcement learning and related ADP methods avoid the direct use of the unknown system dynamics by integrating the Bellman equation over a finite time interval. Based on this idea, policy iteration (PI) and value iteration (VI) algorithms have been developed for continuous-time LQR problems~\cite{Vrabie2009,Jiang2012,Bian2016,BianJiang2022VI,JiangZhou2022BiasPI}. PI-based methods generally learn the optimal solution by solving a sequence of linear equations starting from an initially stabilizing feedback gain, whereas VI-based and modified iteration methods can relax this initialization requirement. Nevertheless, most of these methods are developed under state feedback and require the system state to be measured during learning.

Output feedback learning is essential for practical LQR problems because the full system state is often unavailable. Several RL/ADP-based output feedback methods have been developed for unknown linear systems. Early approaches used measured input--output data or value-function approximation to construct learning equations~\cite{LewisVamvoudakis2011}, while off-policy output-feedback methods were later proposed for unknown continuous-time systems~\cite{Modares2016}. More recent studies have addressed optimal output tracking and regulation under unmeasurable states using dynamic output feedback and input--output data~\cite{Chen2022OutputTracking,Xie2023OutputReg,Zhao2024Incremental,Jiang2025OutputReg}. However, for continuous-time systems, direct input-output parametrization may involve derivatives of input and output signals, which are difficult to obtain accurately in practice. Moreover, exploration signals may introduce estimation bias, leading to additional treatments such as discounted cost functions or bias compensation.

A significant step toward continuous-time model-free output feedback LQR was made in~\cite{Rizvi2020}, where dynamic filters driven by measurable input and output signals were introduced. The resulting filtered input--output vector $z(t)$ can asymptotically parameterize the unmeasured state, so that both the state-feedback value function and the control law can be expressed as functions of $z(t)$. The additional rank conditions required by such state parameterizations were subsequently analyzed in~\cite{RizviLin2023StateParam}. Based on this parametrization, model-free output feedback PI and VI Bellman equations were developed in~\cite{Rizvi2020}. 
This dynamic output feedback framework has also been extended to related problems, such as continuous-time zero-sum games~\cite{RizviLin2020Game}. Recent studies further refined this framework by introducing internal controller states and observer-error compensation to improve the learning equation and reduce computational requirements~\cite{Xie2024}, and by addressing the transient optimality loss caused by observer errors through an input-output-driven internal model~\cite{Xie2025}. More recently, an output-feedback homotopy-based PI method was developed to alleviate the initial stabilizing-policy requirement while providing performance assurance~\cite{Chen2025Homotopy}.

Although the dynamic output feedback framework in~\cite{Rizvi2020} establishes that the filtered input--output vector $z(t)$ contains sufficient information to parameterize the unmeasured state, it does not characterize the intrinsic dimension of this parametrization. State-parametrization sufficiency and intrinsic-dimension characterization are two different questions: the former concerns whether the state can be parameterized after the filter transient vanishes, whereas the latter determines how many independent directions are actually contained in the complete filtered vector and are available to the Bellman regression. If this intrinsic dimension is smaller than the ambient dimension of $z(t)$, the full-vector regression contains deterministic dependencies and may be rank deficient even when the collected input--output data are persistently excited.

Motivated by this observation, this article characterizes the intrinsic dimension of the filter-based state parametrization and develops a corresponding reduction approach for model-free output feedback LQR. The key technical idea is to convert the linear dependence among the filtered signals into a polynomial row-vector space problem. Considering a system with $n$-dimensional states, 
$m$-dimensional input, and $p$-dimensional output, we first analyze the special case of SIMO systems (where 
$m=1$). This polynomial-space characterization shows that the complete filtered vector contains only 
$2n$ independent degrees of freedom. For general MIMO systems, we separate the 
$nm$-dimensional subspace generated by the input-filtered variables from the additional directions contributed by the output-filtered variables, and prove that the latter contributes only 
$n$ independent directions.
Therefore, the intrinsic dimension of the conventional filtered vector is $n(m+1)$, rather than $n(m+p)$. Based on this result, a reduced filtered vector $z_r(t)$ is selected directly from data by column-pivoted QR decomposition, and the model-free output feedback PI and VI equations are reformulated in terms of $z_r(t)$. The proposed reduction removes the deterministic rank deficiency caused by linearly dependent filtered variables, reduces the number of unknown regression parameters, and preserves the fully input-output data-based nature of the learning algorithm.

The rest of this article is organized as follows. Section~II reviews the continuous-time LQR problem and the filter-based dynamic output feedback parametrization used in model-free learning. Section~III characterizes the intrinsic dimension of the filtered input--output parametrization for SIMO and MIMO systems. Section~IV derives the reduced model-free output feedback PI and VI equations. Section~V presents a numerical example based on a jet transport aircraft benchmark system to illustrate the rank reduction and closed-loop performance. Section~VI concludes the article. The proofs of the polynomial representation lemmas are given in the appendices.

\emph{Notation:} For a matrix $X$, $X^{\top}$ denotes its transpose, and $\operatorname{rank}(X)$ denotes its rank. 
$I_n\in \mathbb{R}^{n\times n}$ is a unit matrix, and the zero vector or zero matrix with compatible dimension is denoted by $\textbf{0}$. For a square matrix $A$, $\operatorname{adj}(A)$ denotes its adjugate matrix. The set $\mathbb P^n$ denotes the normed space of all $n$-by-$n$ real symmetric matrices, and $\mathbb P^n_+\triangleq\{P\in\mathbb P^n:P\geq0\}$. The operator $\col(\cdot)$ denotes column-wise stacking. If $A=[a_1,\cdots, a_s]$, then $\operatorname{vec}(A)=\operatorname{col}(a_1,\cdots,a_s)$. The symbol $\otimes$ denotes the Kronecker product. For $x=[x_1,\ldots,x_n]^\top\in\mathbb{R}^n$, define
$\operatorname{vecv}(x)=[x_1^2,x_1x_2,\ldots,x_1x_n,x_2^2,x_2x_3,\ldots,x_{n-1}x_n,x_n^2]^\top.$ For a symmetric matrix $P\in \mathbb{R}^{n\times n}$, define
$\operatorname{vecs}(P)=[p_{11},2p_{12},\ldots,2p_{1n},p_{22},2p_{23},\ldots,2p_{n-1,n},p_{nn}]^\top.$ 
For a polynomial $f(s)$, $\deg f(s)$ denotes its degree. The notation $\gcd(f_1(s),\ldots,f_q(s))$ denotes the greatest common divisor of polynomials $f_1(s),\ldots,f_q(s)$. The set $\mathbb{R}^{p\times m}[s]$ denotes the set of $p\times m$ polynomial matrices with real coefficients. For a set of vectors or polynomial vectors $\mathcal{S}$, $\operatorname{span}\mathcal{S}$ denotes its linear span and $\dim\operatorname{span}\mathcal{S}$ denotes the dimension of this span. For a linear subspace $\mathcal U\subseteq\mathbb R^{1\times m}[s]$ and polynomial row vectors $f(s),g(s)\in\mathbb R^{1\times m}[s]$, the notation $f(s)\equiv g(s)\pmod{\mathcal U}$ means that $f(s)-g(s)\in\mathcal U$; equivalently, $f(s)$ and $g(s)$ represent the same coset in the quotient space $\mathbb R^{1\times m}[s]/\mathcal U$.


\section{Preliminaries}

Consider the continuous-time linear time-invariant system
\begin{equation}
\begin{cases}
    \dot{x}(t)=Ax(t)+Bu(t),\\
    y(t)=Cx(t),
\end{cases}
\label{eq:plant_system}
\end{equation}
where $x(t)\in\R^n$ is the system state, $u(t)\in\R^m$ is the control input, and $y(t)\in\R^p$ is the measured output. We impose the following standard assumptions:

\begin{assumption}
The pair $(A,B)$ is controllable.
\label{ass:controllable}
\end{assumption}

\begin{assumption}
The pair $(C,A)$ is observable.
\label{ass:observable}
\end{assumption}

The LQR problem is to minimize the infinite-horizon quadratic cost
\begin{equation}
    J=\int_0^\infty
    \left(y\T(t)Q_y y(t)+u\T(t)Ru(t)\right)dt,
    \label{eq:lqr_cost}
\end{equation}
where $Q_y=Q_y\T>0$ and $R=R\T>0$ are known, prescribed weighting matrices. Define $Q=C\T Q_y C$. The optimal state-feedback controller is given by
\begin{equation}
    u^*(t)=K^*x(t),
\end{equation}
where
\begin{equation}
    K^*=-R^{-1}B\T P^*,
\end{equation}
and $P^*=(P^*)\T>0$ is the solution of ARE
\begin{equation} 
    A\T P+PA+Q-PBR^{-1}B\T P=0.
    \label{eq:are}
\end{equation}

\subsection{Model-Based and Model-Free LQR Using State Feedback}

Since the ARE is nonlinear, iterative schemes are often used to compute its solution. A classical PI method starts from a stabilizing feedback gain $K_0$~\cite{Kleinman1968}. The procedures of policy evaluation and policy improvement are given by
\begin{subequations}
\begin{align}
    0&=A_k\T P_k+P_k A_k+Q+K_k\T RK_k,
    \label{eq:state_feedback_pi_evaluation}\\
    K_{k+1}&=-R^{-1}B\T P_k.
    \label{eq:state_feedback_pi_improvement}
\end{align}
\end{subequations}
where $A_k=A+BK_k$. This iteration is known to converge to the optimal solution under a stabilizing initial policy. However, the update requires the system matrices $A$ and $B$, as well as the state-weighting matrix $Q$.

Besides PI, VI can also be used to solve the LQR problem~\cite{Bian2016}. Different from PI, which requires an initially stabilizing feedback gain, VI can start from any positive definite matrix $(P_0>0)$. A model-based VI scheme first computes the tentative update
\begin{equation}
\widetilde P_{k+1}=P_k+\epsilon_k\left(A^\top P_k+P_kA+Q-P_kBR^{-1}B^\top P_k\right),
\label{eq:model_based_vi_update}
\end{equation}
where $P_k$ is the estimate of the Riccati solution at the $k$-th iteration, and $\epsilon_k>0$ is a stepsize. The sequence $\{\mathcal B_q\}_{q=0}^{\infty}$ consists of prescribed, norm-bounded truncation regions in $\mathbb P_+^n$ with nonempty interiors. The sets and stepsizes satisfy
\begin{equation}
\begin{aligned}
&\mathcal B_q\subseteq\mathcal B_{q+1},\quad q\in\mathbb N,
\qquad \lim_{q\to\infty}\mathcal B_q=\mathbb P_+^n,\\
&\epsilon_k>0,\qquad
\sum_{k=0}^{\infty}\epsilon_k=\infty,\qquad
\sum_{k=0}^{\infty}\epsilon_k^2<\infty.
\end{aligned}
\label{eq:model_based_vi_conditions}
\end{equation}
A typical choice is $\mathcal B_q=\{P\in\mathbb P_+^n:\|P\|\le b_q\}$, where $0<b_0<b_1<\cdots$ and $b_q\to\infty$. If $\widetilde P_{k+1}\in\mathcal B_q$, the tentative update is accepted; otherwise, the iterate is reset to $P_0$ and the truncation index is increased. The boundedness of each $\mathcal B_q$ prevents excessively large transient iterates, while the limit condition ensures that the unknown Riccati solution is not permanently excluded. Since $P^*\in\mathbb P_+^n$, there exists a finite $q^*$ such that $P^*$ lies in the interior of $\mathcal B_{q^*}$. Hence, for sufficiently large $q$, the truncation mechanism does not alter the local convergence of the VI sequence to $P^*$; see~\cite{Bian2016} for details.

We now recall the model-free state-feedback ADP schemes for solving the LQR problem without using the system matrices $A$ and $B$. In this subsection, the full state $x(t)$ is assumed to be measurable. The purpose is to show how the model-based PI and VI equations can be transformed into data-driven Bellman regression equations~\cite{Jiang2012,Bian2016}.

For a sequence of time instants $t_0,t_1,\ldots,t_l$, define the data stack operators, for compatible signals $a(t)$ and $b(t)$, as
\[
\begin{aligned}
\delta_{aa}
&=\left[\operatorname{vecv}(a(t_1))-\operatorname{vecv}(a(t_0)),\ldots,
\operatorname{vecv}(a(t_l))-\operatorname{vecv}(a(t_{l-1}))\right]\T,\\
\Gamma_{ab}
&=\left[\int_{t_0}^{t_1}a\otimes b\,d\tau,\ldots,
\int_{t_{l-1}}^{t_l}a\otimes b\,d\tau\right]\T,\\
I_{aa}
&=\left[\int_{t_0}^{t_1}\operatorname{vecv}(a)\,d\tau,\ldots,
\int_{t_{l-1}}^{t_l}\operatorname{vecv}(a)\,d\tau\right]\T,\\
I_{au}
&=\left[\int_{t_0}^{t_1}a\otimes Ru\,d\tau,\ldots,
\int_{t_{l-1}}^{t_l}a\otimes Ru\,d\tau\right]\T.
\end{aligned}
\]
In model-free PI, at the $k$-th iteration, rewrite the system~\eqref{eq:plant_system} as
\begin{equation}
\dot{x}=A_kx+B(u-K_kx),
\qquad
A_k=A+BK_k.
\label{eq:policy_shifted_system}
\end{equation}
For a fixed policy $K_k$, combining the PI relations~\eqref{eq:state_feedback_pi_evaluation} and~\eqref{eq:state_feedback_pi_improvement} with~\eqref{eq:policy_shifted_system}, and integrating the derivative of the value function $V_k(x)=x^\top P_kx$ over $[t,t+\Delta t]$, gives
\begin{equation}
\begin{aligned}
&x^\top(t+\Delta t)P_kx(t+\Delta t)-x^\top(t)P_kx(t) \\
&=\int_t^{t+\Delta t}
\big[-x^\top(\tau)(Q+K_k^\top RK_k)x(\tau) \\
&\qquad{}-2\big(u(\tau)-K_kx(\tau)\big)^\top RK_{k+1}x(\tau)\big]d\tau .
\end{aligned}
\label{eq:model_free_pi_integral}
\end{equation}
It can be written in the regression form
\begin{equation}
\Psi_k
\begin{bmatrix}
\operatorname{vecs}(P_k)\\
\operatorname{vec}(K_{k+1})
\end{bmatrix}
=
\Phi_k,
\label{eq:model_free_pi_regression}
\end{equation}
where $\Psi_k=\big[\delta_{xx},\ -2\Gamma_{xx}\big(I_n\otimes K_k\T R\big)+2\Gamma_{xu}\big]$ and $\Phi_k=-\Gamma_{xx}\operatorname{vecs}\big(Q+K_k\T RK_k\big)$. Hence, $P_k$ and $K_{k+1}$ can be obtained by least squares without explicitly using $A$ and $B$. Repeating this procedure yields a sequence of feedback gains that converges to the optimal state-feedback gain under the standard rank condition and an initially stabilizing policy.

The state-feedback PI method removes the need for system matrices, but it still requires a stabilizing initial gain $K_0$. To avoid this requirement, model-free VI can be used. Starting from any positive definite matrix $P_0>0$, define $H_k=A^\top P_k+P_kA$ and $K_k=-R^{-1}B^\top P_k$, along the system trajectory, we obtain
\begin{equation}
\begin{aligned}
&x^\top(t+\Delta t)P_kx(t+\Delta t)-x^\top(t)P_kx(t)\\
&=\int_t^{t+\Delta t}
\big[x^\top(\tau)H_kx(\tau)
-2u^\top(\tau)RK_kx(\tau)\big] \,d\tau .
\end{aligned}
\label{eq:model_free_vi_integral}
\end{equation}
After stacking the data over all time intervals, one obtains
\begin{equation}
\Omega
\begin{bmatrix}
\operatorname{vecs}(H_k)\\
\operatorname{vec}(K_k)
\end{bmatrix}
=
\Theta_k,
\label{eq:model_free_vi_regression}
\end{equation}
where $\Omega=\begin{bmatrix} I_{xx} & -2I_{xu}\end{bmatrix}$ and $\Theta_k=\delta_{xx}\operatorname{vecs}(P_k)$. Solving this linear regression equation gives $H_k$ and $K_k$ from state-input data. The value matrix is then updated by
\begin{equation}
P_{k+1}=P_k+\epsilon_k\left(H_k+Q-K_k^\top RK_k\right),
\label{eq:model_free_vi_update}
\end{equation}
where the stepsize and bounded-set sequence satisfy the VI convergence conditions in~\eqref{eq:model_based_vi_conditions}. Therefore, model-free VI does not require the system matrices $(A,B)$, nor does it require an initially stabilizing feedback gain.

\subsection{Dynamic Output Feedback State Parametrization}

The above model-free PI and VI schemes provide the basis for continuous-time LQR learning. However, they both rely on the availability of the full state $x(t)$ in the data stack operators and in the Bellman regression equations. In many practical systems, the internal state is not measurable and only the input-output data $u(t)$ and $y(t)$ are available.

To overcome this difficulty, the dynamic output feedback learning method in~\cite{Rizvi2020} introduces a filtered input--output vector to parameterize the unmeasured state. We recall this construction because it provides the notation and channelwise filtered-variable parametrization needed in the subsequent dependence analysis and reduced formulation.

Let
\begin{equation}
    \Lambda(s)=s^n+\alpha_{n-1}s^{n-1}+\cdots+
    \alpha_1s+\alpha_0
    \label{eq:stable_filter_polynomial}
\end{equation}
be a user-selected Hurwitz polynomial. Let $A_f\in\R^{n\times n}$ be the companion matrix associated with $\Lambda(s)$, and let $b_f=\begin{bmatrix}0&0&\cdots&0&1\end{bmatrix}\T$. For each input component $u_i(t)$, $i=1,\ldots,m$, and each output component $y_j(t)$, $j=1,\ldots,p$, define the filtered variables
\begin{subequations}
\label{eq:z_system}
\begin{align}
    \dot{\zeta}_{u_i}(t)&=A_f\zeta_{u_i}(t)+b_fu_i(t),
    &\zeta_{u_i}(t)&\in\R^n,
    \label{eq:input_filter_dynamics}\\
    \dot{\zeta}_{y_j}(t)&=A_f\zeta_{y_j}(t)+b_fy_j(t),
    &\zeta_{y_j}(t)&\in\R^n.
    \label{eq:output_filter_dynamics}
\end{align}
\end{subequations}
Collect them as $\zeta_u(t)=\col\big(\zeta_{u_1}(t),\ldots,\zeta_{u_m}(t)\big)\in\R^{nm}$ and $\zeta_y(t)=\col\big(\zeta_{y_1}(t),\ldots,\zeta_{y_p}(t)\big)\in\R^{np}$, and the full filtered input--output vector is defined by
\begin{equation}
    z(t)=\col\big(\zeta_u(t),\zeta_y(t)\big)
    \in\R^{n(m+p)}.
    \label{eq:z_full}
\end{equation}

The following state-parametrization lemma and its observer-based proof are recalled from~\cite{Rizvi2020}, rather than presented as a new contribution. The proof is retained for completeness and to make explicit the filtered variables and mapping matrices used in the subsequent proofs.
\begin{lemma}
Consider system~\eqref{eq:plant_system}, under Assumption~\ref{ass:observable}, there exist constant matrices $M_u\in\R^{n\times nm}$ and $M_y\in\R^{n\times np}$ such that
\begin{equation}
    x(t)=M_u\zeta_u(t)+M_y\zeta_y(t)+\varepsilon(t),
    \label{eq:state_param_original}
\end{equation}
where $\varepsilon(t)$ is a transient term satisfying $\lim_{t\to\infty}\varepsilon(t)=0$. Equivalently, with $M=\begin{bmatrix}M_u&M_y\end{bmatrix}$, the state can be written in terms of the full filtered vector $z(t)$. Therefore, after the filter transient vanishes, the state can be asymptotically parameterized by
\begin{equation}
    x(t)\approx Mz(t).
    \label{eq:full_state_parametrization}
\end{equation}
\label{lem:io_state_parametrization}
\end{lemma}

\begin{proof}
Following the observer-based construction in~\cite{Rizvi2020}, under Assumption~\ref{ass:observable}, there exists an observer gain $L\in\R^{n\times p}$ such that $A_L=A-LC$ is Hurwitz and has the characteristic polynomial $\Lambda(s)$. Consider the corresponding Luenberger observer dynamics
\begin{equation}
    \dot{\hat{x}}(t)=A_L\hat{x}(t)+Bu(t)+Ly(t).
    \label{eq:observer_equivalent_dynamics}
\end{equation}
Taking the Laplace transform of \eqref{eq:observer_equivalent_dynamics} and setting the observer initial condition to zero, we have
\begin{equation}
\begin{aligned}
    \hat{x}(t)
    &={\mathcal L}^{-1}\!\Bigg\{
    (sI-A_L)^{-1}BU(s)
    +(sI-A_L)^{-1}LY(s)
    \Bigg\} \\
    &={\mathcal L}^{-1}\!\Bigg\{
    \sum_{i=1}^{m}\frac{\operatorname{adj}(sI-A_L)B_iU_i(s)}{\Lambda(s)}
    +
    \sum_{j=1}^{p}\frac{\operatorname{adj}(sI-A_L)L_jY_j(s)}{\Lambda(s)}
    \Bigg\},
\end{aligned}
    \label{eq:observer_laplace_decomposition}
\end{equation}
where $B_i$ and $L_j$ denote the $i$th column of $B$ and the $j$th column of $L$, respectively. Consider the input-channel term
\begin{equation}
\begin{aligned}
    \frac{\operatorname{adj}(sI-A_L)B_i}{\Lambda(s)}U_i(s)
    &=
    \begin{bmatrix}
    \sum_{k=0}^{n-1}a_k^{1i}s^k\\
    \vdots\\
    \sum_{k=0}^{n-1}a_k^{ni}s^k
    \end{bmatrix}
    \frac{U_i(s)}{\Lambda(s)} \\
    &=
    \begin{bmatrix}
    a_0^{1i}&\cdots&a_{n-1}^{1i}\\
    \vdots&\ddots&\vdots\\
    a_0^{ni}&\cdots&a_{n-1}^{ni}
    \end{bmatrix}
    \begin{bmatrix}
    1\\ \vdots\\ s^{n-1}
    \end{bmatrix}
    \frac{U_i(s)}{\Lambda(s)} \\
    &=M_{u_i}\zeta_{u_i}(s).
\end{aligned}
\end{equation}
Here $\zeta_{u_i}(s)=\begin{bmatrix}1&\cdots&s^{n-1}\end{bmatrix}\T U_i(s)/\Lambda(s)$ is the Laplace-domain filtered signal generated by \eqref{eq:input_filter_dynamics}.
Similarly, for each output channel, there exists a constant matrix $M_{y_j}\in\R^{n\times n}$ such that
\begin{equation}
    {\mathcal L}^{-1}\!\left\{
    \frac{\operatorname{adj}(sI-A_L)L_j}{\Lambda(s)}Y_j(s)
    \right\}
    =M_{y_j}\zeta_{y_j}(t).
\end{equation}
Therefore,
\begin{equation}
    \hat{x}(t)=M_u\zeta_u(t)+M_y\zeta_y(t),
\end{equation}
where $M_u=\begin{bmatrix}M_{u_1}&\cdots&M_{u_m}\end{bmatrix}$ and $M_y=\begin{bmatrix}M_{y_1}&\cdots&M_{y_p}\end{bmatrix}$. Since $x(t)=\hat{x}(t)+\varepsilon(t)$, where $\varepsilon(t)$ is the observer error satisfying $\varepsilon(t)\to0$ because $A_L$ is Hurwitz, this proves \eqref{eq:state_param_original} and completes the proof.
\end{proof}

\subsection{Model-Free LQR Using Output Feedback}

Following the dynamic output feedback framework in~\cite{Rizvi2020}, this subsection recalls how the state-feedback Bellman equations are converted into filtered input--output equations. Let
$n_z=n(m+p)$ and define
$\bar P_k=M\T P_kM$, $\bar K_k=K_kM$, and
$\bar H_k=M\T H_kM$. Substituting $x(t)\approx Mz(t)$ into
\eqref{eq:model_free_pi_integral} gives the full filtered-vector PI equation
\begin{equation}
\begin{aligned}
&z\T(t+\Delta t)\bar P_kz(t+\Delta t)-z\T(t)\bar P_kz(t)\\
&=-\int_t^{t+\Delta t}\!\Big[y\T(\tau)Q_y y(\tau)
+z\T(\tau)\bar K_k\T R\bar K_kz(\tau)\Big]d\tau\\
&\quad{}-2\int_t^{t+\Delta t}\!\big(u(\tau)-\bar K_kz(\tau)\big)\T
R\bar K_{k+1}z(\tau)\,d\tau .
\end{aligned}
\label{eq:full_filtered_pi_integral}
\end{equation}
Here the state-cost term has been written as $x\T Qx=y\T Q_y y$, so the equation only uses measured output data and the filtered vector $z(t)$. Stacking \eqref{eq:full_filtered_pi_integral} over all data intervals yields
\begin{equation}
\bar{\Psi}_k
\begin{bmatrix}
\operatorname{vecs}(\bar P_k)\\
\operatorname{vec}(\bar K_{k+1})
\end{bmatrix}
=
\bar{\Phi}_k,
\label{eq:full_filtered_pi_regression}
\end{equation}
where $\bar{\Psi}_k=\big[\delta_{zz},\ -2\Gamma_{zz}\big(I_{n_z}\otimes \bar K_k\T R\big)+2\Gamma_{zu}\big]$ and $\bar{\Phi}_k=-I_{yy}\operatorname{vecs}(Q_y)-\Gamma_{zz}\operatorname{vecs}\big(\bar K_k\T R\bar K_k\big)$. 
Similarly, substituting $x(t)\approx Mz(t)$ into \eqref{eq:model_free_vi_integral} gives
\begin{equation}
\begin{aligned}
&z\T(t+\Delta t)\bar P_kz(t+\Delta t)-z\T(t)\bar P_kz(t)\\
&=\int_t^{t+\Delta t}
\big[z\T(\tau)\bar H_kz(\tau)-2u\T(\tau)R\bar K_kz(\tau)\big] \,d\tau .
\end{aligned}
\label{eq:full_filtered_vi_integral}
\end{equation}
The corresponding full filtered-vector VI regression is
\begin{equation}
\bar{\Omega}
\begin{bmatrix}
\operatorname{vecs}(\bar H_k)\\
\operatorname{vec}(\bar K_k)
\end{bmatrix}
=
\bar{\Theta}_k,
\label{eq:full_filtered_vi_regression}
\end{equation}
where $\bar{\Omega}=\begin{bmatrix}I_{zz}&-2I_{zu}\end{bmatrix}$ and $\bar{\Theta}_k=\delta_{zz}\operatorname{vecs}(\bar P_k)$. Thus, the unmeasured state has been eliminated from both the PI and VI Bellman regressions. Existing filter-based output feedback learning therefore replaces the unknown state by the filtered vector $z(t)$ and transforms the Bellman equations into input--output regressions. However, the full-vector formulation implicitly regards all components of $z(t)$ as independent regressors. Since these filtered variables are generated through common system dynamics, this assumption is not obvious. The following section determines the intrinsic dimension of the filtered parametrization and shows that the complete filtered vector contains deterministic redundancy.

\section{Intrinsic Dimension of Filtered Input--Output Signals}

\subsection{Problem Formulation}
The full filtered-vector PI and VI regressions in~\eqref{eq:full_filtered_pi_regression} and~\eqref{eq:full_filtered_vi_regression} use the complete vector $z(t)\in\R^{n(m+p)}$ as the state parameterization and implicitly treat all its components as independent regression directions. However, because the filtered input and output variables are generated through common system dynamics, the number of independent directions in $z(t)$ need not equal its ambient dimension $n(m+p)$.

The objective of this section is therefore not merely to identify linear dependence among the filtered variables, but to determine the intrinsic dimension of the filter-based state parametrization, namely, the number of independent directions contained in the complete filtered vector. We first study the SIMO case and then extend the analysis to general MIMO systems. The main results show that the intrinsic dimension is $2n$ for SIMO systems and $n(m+1)$ for general MIMO systems.

\begin{breakableproblem}{Intrinsic Dimension of the Filtered Representation}
Given the filtered input--output vector $z(t)\in\R^{n(m+p)}$, characterize the dimension of the linear space spanned by its scalar components, and establish how this intrinsic dimension provides the theoretical basis for constructing a reduced Bellman regression.
\label{prob:filtered_intrinsic_dimension}
\end{breakableproblem}

\subsection{SIMO Systems}
We first consider the SIMO case, namely $m=1$. In this case, the full filtered vector is
\begin{equation}
    z(t)=\col\big(\zeta_u(t),\zeta_{y_1}(t),\ldots,\zeta_{y_p}(t)\big)
    \in\R^{n(1+p)}.
    \label{eq:simo_full_filtered_vector}
\end{equation}

In the Laplace domain, the $i$th output can be written as
\begin{equation}
    Y_i(s)=C_i(sI-A)^{-1}BU(s)=\frac{d_i(s)}{a(s)}U(s),
    \qquad i=1,\ldots,p,
\label{eq:simo_io_relation}
\end{equation}
where $a(s)=\det(sI-A)$ is the characteristic polynomial of $A$, and $d_i(s)=C_i\operatorname{adj}(sI-A)B$, with $C_i$ being the $j$th row of $C$. From the filter definition in~\eqref{eq:input_filter_dynamics} and~\eqref{eq:output_filter_dynamics} of the previous section, we obtain $(sI-A_f)^{-1}b_f=\begin{bmatrix}
    1&s&\cdots&s^{n-1}
\end{bmatrix}\frac{1}{\Lambda(s)}$, so the filtered vectors can be written in the Laplace domain as
\begin{equation}
    \zeta_u(s)=
    \begin{bmatrix}1&s&\cdots&s^{n-1}\end{bmatrix}\T\frac{U(s)}{\Lambda(s)},
    \qquad
    \zeta_{y_i}(s)=
    \begin{bmatrix}1&s&\cdots&s^{n-1}\end{bmatrix}\T\frac{Y_i(s)}{\Lambda(s)}.
\label{eq:simo_filter_laplace}
\end{equation}
For $r=0,\ldots,n-1$, the components in \eqref{eq:simo_filter_laplace} associated with the power $s^r$ are $s^rU(s)/\Lambda(s)$ and $s^rY_i(s)/\Lambda(s)$, respectively.

Although its dimension is $n(1+p)$, all output-filtered variables are generated by the same input and the same system dynamics. Therefore, they may contain deterministic linear dependence. We discuss two cases.

\medskip
\subsubsection{Case I: An Observable Pair $(C_i,A)$ Exists}
\smallskip
Assume that there exists an index $i\in\{1,\ldots,p\}$ such that the pair $(C_i,A)$ is observable, so the SISO transfer function from $u$ to $y_i$ has no pole-zero cancellation. Equivalently,
\begin{equation}
    \gcd(a(s),d_i(s))=1.
    \label{eq:simo_single_channel_coprime}
\end{equation}

\begin{lemma}Suppose that $\gcd(a(s),d_i(s))=1$, where $\deg a(s)=n$ and $\deg d_i(s)\leq n-1$. Then for any $j\neq i$ and any $r=0,\ldots,n-1$, there exist polynomials $\alpha_{jr}(s)$ and $\beta_{jr}(s)$ such that
\begin{equation}
    \alpha_{jr}(s)a(s)+\beta_{jr}(s)d_i(s)=s^rd_j(s),
    \label{eq:simo_case_i_lemma_identity}
\end{equation}
where $\deg\alpha_{jr}(s)<n$ and $\deg\beta_{jr}(s)<n$.
\label{lem:simo_case_i_representation}
\end{lemma}

The proof is given in Appendix~\ref{app:simo_case_i_lemma_proof}.

\begin{proposition}If there exists an index $i\in\{1,\ldots,p\}$ such that the pair $(C_i,A)$ is observable, then every component of $\zeta_{y_j}(t)$, $j\neq i$, belongs to the span of the components of $\zeta_u(t)$ and $\zeta_{y_i}(t)$. Moreover, the $2n$ components of $\zeta_u(t)$ and $\zeta_{y_i}(t)$ are linearly independent. Consequently, the space spanned by the components of the full SIMO filtered vector has dimension $2n$.
\label{prop:simo_case_i_reduction}
\end{proposition}

\begin{proof}
The proof consists of two steps.

\smallskip
\indent\textbf{Step 1: Representation of linearly dependent output-filtered components.}\quad
Lemma~\ref{lem:simo_case_i_representation} is used to show that each component of $\zeta_{y_j}(t)$, $j\neq i$, can be written as a linear combination of the components of $\zeta_u(t)$ and $\zeta_{y_i}(t)$. Indeed, multiplying \eqref{eq:simo_case_i_lemma_identity} by $U(s)/(a(s)\Lambda(s))$ and using $Y_\ell(s)=d_\ell(s)U(s)/a(s)$, $\ell=i,j$, gives
\begin{equation}
    \frac{s^rY_j(s)}{\Lambda(s)}
    =\frac{\alpha_{jr}(s)U(s)}{\Lambda(s)}
    +\frac{\beta_{jr}(s)Y_i(s)}{\Lambda(s)}.
    \label{eq:simo_case_i_filter_identity}
\end{equation}
Define
\begin{equation*}
    \alpha_{jr}(s)=\alpha_{jr,0}+\alpha_{jr,1}s+
    \cdots+\alpha_{jr,n-1}s^{n-1},
    \qquad
    \beta_{jr}(s)=\beta_{jr,0}+\beta_{jr,1}s+
    \cdots+\beta_{jr,n-1}s^{n-1},
\end{equation*}
and let $A_{jr}\triangleq[\alpha_{jr,0}\ \alpha_{jr,1}\ \cdots\ \alpha_{jr,n-1}]$ and $B_{jr}\triangleq[\beta_{jr,0}\ \beta_{jr,1}\ \cdots\ \beta_{jr,n-1}]$. Combining \eqref{eq:simo_filter_laplace} and \eqref{eq:simo_case_i_filter_identity} implies
\begin{equation*}
    [\zeta_{y_j}(s)]_{r+1}
    =A_{jr}\zeta_u(s)+B_{jr}\zeta_{y_i}(s).
\end{equation*}
Equivalently, in the time domain,
\begin{equation}
    [\zeta_{y_j}(t)]_{r+1}
    =A_{jr}\zeta_u(t)+B_{jr}\zeta_{y_i}(t).
    \label{eq:simo_case_i_component_time}
\end{equation}
Since \eqref{eq:simo_case_i_component_time} holds for every $r=0,\ldots,n-1$, each component of $\zeta_{y_j}(t)$, $j\neq i$, can be written as a linear combination of the components of $\zeta_u(t)$ and $\zeta_{y_i}(t)$.

\smallskip
\indent\textbf{Step 2: Linear independence of the reduced components.}\quad
We prove that $\zeta_u(t)$ and $\zeta_{y_i}(t)$ together form $2n$ linearly independent components. Equivalently, we need to show that the polynomials $a(s),sa(s),\ldots,s^{n-1}a(s),d_i(s),sd_i(s),\ldots,s^{n-1}d_i(s)$ are linearly independent. Suppose, to the contrary, that there exists a nonzero pair of row vectors $v\in\R^{1\times n}$ and $w\in\R^{1\times n}$ such that
\begin{equation}
    v
    \begin{bmatrix}
    a(s)\\
    sa(s)\\
    \vdots\\
    s^{n-1}a(s)
    \end{bmatrix}
    +w
    \begin{bmatrix}
    d_i(s)\\
    sd_i(s)\\
    \vdots\\
    s^{n-1}d_i(s)
    \end{bmatrix}
    =0.
    \label{eq:simo_case_i_linear_dependence_contradiction}
\end{equation}
Define $r(s)=v[1\ s\ \cdots\ s^{n-1}]\T$ and $q(s)=-w[1\ s\ \cdots\ s^{n-1}]\T$. Then $\deg r(s)\leq n-1$, $\deg q(s)\leq n-1$, and \eqref{eq:simo_case_i_linear_dependence_contradiction} can be rewritten as
\begin{equation}
    r(s)a(s)=q(s)d_i(s).
    \label{eq:simo_case_i_coprime_product_identity}
\end{equation}
Since $\gcd(a(s),d_i(s))=1$, the polynomial $a(s)$ must divide $q(s)$. However, $\deg q(s)\leq n-1$ and $\deg a(s)=n$, which implies $q(s)=0$. Substituting this into \eqref{eq:simo_case_i_coprime_product_identity} gives $r(s)a(s)=0$, and hence $r(s)=0$. From the definitions of $r(s)$ and $q(s)$, it follows that $v=0$ and $w=0$, contradicting the assumption that $(v,w)$ is nonzero. Hence, the above $2n$ polynomials are linearly independent, and therefore
\begin{equation}
    \dim\spanop\{a(s),sa(s),\ldots,s^{n-1}a(s),d_i(s),sd_i(s),\ldots,s^{n-1}d_i(s)\}=2n.
    \label{eq:simo_case_i_independent_set_dimension}
\end{equation}
Consequently, $\zeta_u(t)$ and $\zeta_{y_i}(t)$ contain $2n$ linearly independent components.
\end{proof}

Proposition~\ref{prop:simo_case_i_reduction} shows that the full SIMO filtered vector need not be retained in Case I. A reduced vector may be chosen as
\begin{equation}
    z_r(t)=\col\big(\zeta_u(t),\zeta_{y_i}(t)\big)\in\R^{2n}.
    \label{eq:simo_case_i_reduced_vector}
\end{equation}

\medskip
\subsubsection{Case II: No Observable Pair $(C_i,A)$ Exists}
\smallskip
For each $i=1,\ldots,p$, the pair $(C_i,A)$ is not observable. Hence, the transfer function $d_i(s)/a(s)$ has pole-zero cancellation, that is, $\gcd(a(s),d_i(s))\neq1$ for all $i=1,\ldots,p$.
However, under Assumption~\ref{ass:observable}, the output channels collectively contain all state information. This implies that the transfer function vector $G(s)=\frac{1}{a(s)}\begin{bmatrix} d_1(s) & d_2(s) &\cdots & d_p(s)\\
\end{bmatrix}\T$ has no common pole-zero cancellation.
Equivalently,
\begin{equation}
    \gcd(a(s),d_1(s),d_2(s),\ldots,d_p(s))=1.
    \label{eq:simo_joint_coprime}
\end{equation}

\begin{proposition}
For each output channel $y_i$, let $\rho_i(s)=\gcd(a(s),d_i(s))$,
and define the coprime denominator and numerator as
\begin{equation}
    \bar a_i(s)=\frac{a(s)}{\rho_i(s)},
    \qquad
    \bar d_i(s)=\frac{d_i(s)}{\rho_i(s)}.
    \label{eq:simo_case_ii_reduced_tf_definition}
\end{equation}
Let $\nu_i=\deg\bar a_i(s)<n$, the following two types of deterministic linear dependence hold.

First, within each output-filtered vector $\zeta_{y_i}(t)$, the higher-order components $[\zeta_{y_i}(t)]_{\nu_i+1},\ldots,[\zeta_{y_i}(t)]_n$ can be linearly dependent on $\zeta_u(t)$ and the first $\nu_i$ components of $\zeta_{y_i}(t)$. Therefore, only the first $\nu_i$ components of $\zeta_{y_i}(t)$ may contribute new independent directions beyond the input-filtered vector $\zeta_u(t)$.

Second, different output-filtered vectors  $\zeta_{y_i}(t)$ and $\zeta_{y_j}(t)$, $(j\neq i)$ are also linearly dependent. Thus, the independent output-filtered vector components in Case II are generally distributed among multiple output channels, rather than being contained in a single $\zeta_{y_i}(t)$.
\label{prop:simo_case_ii_dependence_structure}
\end{proposition}

\begin{proof}
The first statement follows from
\begin{equation}
    Y_i(s)=\frac{d_i(s)}{a(s)}U(s)
    =\frac{\bar d_i(s)}{\bar a_i(s)}U(s),
\end{equation}
which gives
\begin{equation}
    \bar a_i(s)Y_i(s)=\bar d_i(s)U(s).
    \label{eq:simo_case_ii_reduced_io_identity}
\end{equation}
Choose $\rho_i(s)$ to be monic, since $\deg\bar d_i(s)\leq\nu_i-1$, the polynomials $\bar a_i(s)$ and $\bar d_i(s)$ can be written as
\begin{equation}
    \bar a_i(s)=s^{\nu_i}+\gamma_{i,\nu_i-1}s^{\nu_i-1}
    +\cdots+\gamma_{i,0},
    \qquad
    \bar d_i(s)=\delta_{i,\nu_i-1}s^{\nu_i-1}
    +\cdots+\delta_{i,0}.
\end{equation}
Dividing \eqref{eq:simo_case_ii_reduced_io_identity} by $\Lambda(s)$ on the both sides gives
\begin{equation}
    [\zeta_{y_i}(s)]_{\nu_i+1}
    =-\gamma_{i,\nu_i-1}[\zeta_{y_i}(s)]_{\nu_i}-\cdots-\gamma_{i,0}[\zeta_{y_i}(s)]_1
    +\delta_{i,\nu_i-1}[\zeta_u(s)]_{\nu_i}+\cdots+\delta_{i,0}[\zeta_u(s)]_1.
    \label{eq:simo_case_ii_internal_first_relation}
\end{equation}
Multiplying \eqref{eq:simo_case_ii_reduced_io_identity} by $s^k/\Lambda(s)$, $k=1,\ldots,n-\nu_i-1$, gives the same type of relation for the remaining higher-order components. Recursively, $[\zeta_{y_i}(t)]_{\nu_i+1},\ldots,[\zeta_{y_i}(t)]_n$ can be represented by $\zeta_u(t)$ and $[\zeta_{y_i}(t)]_1,\ldots,[\zeta_{y_i}(t)]_{\nu_i}$.

For the second statement, since
\begin{equation}
    Y_i(s)=\frac{d_i(s)}{a(s)}U(s),
    \qquad
    Y_j(s)=\frac{d_j(s)}{a(s)}U(s),
\end{equation}
we have
\begin{equation}
    d_j(s)Y_i(s)=d_i(s)Y_j(s).
    \label{eq:simo_case_ii_channel_cross_identity}
\end{equation}
Write
\begin{equation*}
    d_i(s)=d_{i,0}+d_{i,1}s+\cdots+d_{i,n-1}s^{n-1},
    \qquad
    d_j(s)=d_{j,0}+d_{j,1}s+\cdots+d_{j,n-1}s^{n-1},
\end{equation*}
dividing \eqref{eq:simo_case_ii_channel_cross_identity} by $\Lambda(s)$ and using the Laplace-domain filtered-output vectors components in \eqref{eq:simo_filter_laplace}, we obtain
\begin{equation}
    \sum_{\ell=0}^{n-1}d_{j,\ell}[\zeta_{y_i}(s)]_{\ell+1}
    -\sum_{\ell=0}^{n-1}d_{i,\ell}[\zeta_{y_j}(s)]_{\ell+1}=0.
    \label{eq:simo_case_ii_filtered_cross_relation}
\end{equation}
Thus, the components of $\zeta_{y_i}(s)$ and $\zeta_{y_j}(s)$, and equivalently those of $\zeta_{y_i}(t)$ and $\zeta_{y_j}(t)$ in the time domain, satisfy a deterministic linear dependence. This completes the proof.
\end{proof}

Proposition~\ref{prop:simo_case_ii_dependence_structure} shows that, when no individual pair $(C_i,A)$ is observable, the independent output-filtered components cannot generally be obtained by retaining one complete output-filtered vector. Instead, they may be distributed among multiple output channels. Nevertheless, Assumptions~\ref{ass:controllable} and~\ref{ass:observable} guarantee that an observable scalar virtual output can be formed as a linear combination of the measured outputs.

\begin{lemma}
Consider a SIMO system satisfying Assumptions~\ref{ass:controllable} and~\ref{ass:observable}. Then there exists a real vector $v\in\R^p$ such that $(v\T C,A)$ is observable.
\label{lem:simo_observable_output_combination}
\end{lemma}

\begin{proof}
Partition $C$ by rows as $C=\col(C_1\T,\ldots,C_p\T)$, where $C_i\in\R^n$. Since the system is controllable and observable, $(A,B,C)$ is a minimal realization. Its transfer-function vector is
\begin{equation}
    G(s)=\frac{1}{a(s)}
    \begin{bmatrix}d_1(s)&d_2(s)&\cdots&d_p(s)\end{bmatrix}\T,
\end{equation}
which has no common pole-zero cancellation.

Let $\lambda_1,\ldots,\lambda_s$ be the distinct roots of $a(s)$ and define
\begin{equation}
    N_i=\begin{bmatrix}
        d_i(\lambda_1)&d_i(\lambda_2)&\cdots&d_i(\lambda_s)
    \end{bmatrix}\T,\qquad i=1,\ldots,p.
\end{equation}
Because $G(s)$ has no common pole-zero cancellation, for each $j=1,\ldots,s$, at least one of $d_1(\lambda_j),\ldots,d_p(\lambda_j)$ is nonzero.

Without loss of generality, after reordering the rows of $C$ and the roots $\lambda_j$, assume that the first $\kappa_1>0$ entries of $N_1$ are nonzero and its remaining entries are zero. If $\kappa_1=s$, then $d_1(s)/a(s)$ has no pole-zero cancellation, and we can simply choose
\begin{equation}
    v=\begin{bmatrix}1&0&\cdots&0\end{bmatrix}\T.
\end{equation}

Suppose that $\kappa_1<s$. Then there exists some $i>1$ such that the $(\kappa_1+1)$th entry of $N_i$ is nonzero; otherwise, all the numerators would vanish at $\lambda_{\kappa_1+1}$. After reordering $N_2,\ldots,N_p$ if necessary, assume that this vector is $N_2$. By further reordering the remaining roots, assume that $N_2$ has nonzero entries from positions $\kappa_1+1$ to $\kappa_2$ and zero entries after position $\kappa_2$. If $\kappa_2<s$, the same argument can be repeated. Since the numerator vectors have no common zero, this process eventually yields some $\ell\leq p$ such that $\kappa_\ell=s$.

It remains to choose the coefficients of the output combination. There are only finitely many real values of $a_2$ for which $N_1+a_2N_2$ has a zero entry among its first $\kappa_2$ entries. Choose a nonzero real $a_2$ avoiding these values. After $a_2$ has been chosen, there are only finitely many real values of $a_3$ for which $N_1+a_2N_2+a_3N_3$ has a zero entry among its first $\kappa_3$ entries. Choose a nonzero real $a_3$ avoiding these values. Continuing in this way, we obtain nonzero real numbers $a_2,\ldots,a_\ell$ such that
\begin{equation}
    N_1+a_2N_2+\cdots+a_\ell N_\ell
\end{equation}
has no zero entry. Therefore, choose
\begin{equation}
    v=\begin{bmatrix}1&a_2&\cdots&a_\ell&0&\cdots&0\end{bmatrix}\T,
\end{equation}
with its entries returned to the original output order if the rows of $C$ were reordered. Then
\begin{equation}
    v\T C=C_1+a_2C_2+\cdots+a_\ell C_\ell,
\end{equation}
and the corresponding scalar transfer function is
\begin{equation}
    v\T C(sI-A)^{-1}B
    =\frac{d_1(s)+a_2d_2(s)+\cdots+a_\ell d_\ell(s)}{a(s)}.
\end{equation}
Its numerator is nonzero at every root of $a(s)$, so the transfer function has no pole-zero cancellation. Since $(A,B)$ is controllable, it follows that $(v\T C,A)$ is observable.
\end{proof}

Let $v$ be a vector whose existence is guaranteed by Lemma~\ref{lem:simo_observable_output_combination}, and define the virtual scalar output
\begin{equation}
    \widetilde y(t)=v\T y(t)=v\T Cx(t)
    \label{eq:simo_virtual_output}
\end{equation}
and its numerator polynomial
\begin{equation}
    \widetilde d(s)
    =v\T\begin{bmatrix}d_1(s)&\cdots&d_p(s)\end{bmatrix}\T
    =\sum_{i=1}^p v_i d_i(s).
    \label{eq:simo_virtual_numerator}
\end{equation}
The construction in the proof of Lemma~\ref{lem:simo_observable_output_combination} shows that the scalar transfer function $\widetilde d(s)/a(s)$ has no pole-zero cancellation. Therefore,
\begin{equation}
    \gcd\big(a(s),\widetilde d(s)\big)=1.
    \label{eq:simo_virtual_coprime}
\end{equation}
Moreover, the linearity of the filter in~\eqref{eq:output_filter_dynamics} gives
\begin{equation}
    \zeta_{\widetilde y}(t)=\sum_{i=1}^p v_i\zeta_{y_i}(t).
    \label{eq:simo_virtual_filtered_output}
\end{equation}

\begin{theorem}
Consider a SIMO system satisfying Assumptions~\ref{ass:controllable} and~\ref{ass:observable}. Define
\begin{equation}
    \mathcal S_{\rm SIMO}=\{s^ra(s):r=0,\ldots,n-1\}
    \cup\{s^rd_i(s):r=0,\ldots,n-1,\ i=1,\ldots,p\}.
    \label{eq:simo_polynomial_set}
\end{equation}
Then there exist a vector $v\in\R^p$ and the polynomial $\widetilde d(s)$ in~\eqref{eq:simo_virtual_numerator} such that
\begin{equation}
    \spanop\mathcal S_{\rm SIMO}
    =\spanop\{s^ra(s),s^r\widetilde d(s):r=0,\ldots,n-1\}.
    \label{eq:simo_virtual_span}
\end{equation}
Consequently,
\begin{equation}
    \dim\spanop\mathcal S_{\rm SIMO}=2n.
    \label{eq:simo_span_dimension}
\end{equation}
Equivalently, the complete SIMO filtered vector $z(t)\in\R^{n(1+p)}$ has exactly $2n$ linearly independent components.
\label{thm:simo_filtered_dimension}
\end{theorem}

\begin{proof}
Choose $v\in\R^p$ as in Lemma~\ref{lem:simo_observable_output_combination} and define $\widetilde d(s)$ by~\eqref{eq:simo_virtual_numerator}. Then~\eqref{eq:simo_virtual_coprime} holds. For every $i=1,\ldots,p$ and $r=0,\ldots,n-1$, applying Lemma~\ref{lem:simo_case_i_representation} with $d_i(s)$ replaced by $\widetilde d(s)$ gives polynomials $\alpha_{ir}(s)$ and $\beta_{ir}(s)$, both of degree less than $n$, such that
\begin{equation}
    \alpha_{ir}(s)a(s)+\beta_{ir}(s)\widetilde d(s)
    =s^rd_i(s).
\end{equation}
Therefore,
\begin{equation}
    \spanop\mathcal S_{\rm SIMO}
    \subseteq\spanop\{s^ra(s),s^r\widetilde d(s):
    r=0,\ldots,n-1\}.
\end{equation}
Conversely, since $s^r\widetilde d(s)=\sum_{i=1}^p v_i s^rd_i(s)$, every polynomial generated by the virtual output belongs to $\spanop\mathcal S_{\rm SIMO}$. Hence,~\eqref{eq:simo_virtual_span} holds.

It remains to determine the dimension of this space. Suppose that polynomials $q(s)$ and $h(s)$ satisfy $\deg q(s)<n$, $\deg h(s)<n$, and
\begin{equation}
    q(s)a(s)+h(s)\widetilde d(s)=0.
\end{equation}
Since $\gcd(a(s),\widetilde d(s))=1$, $a(s)$ must divide $h(s)$. However, $\deg h(s)<n=\deg a(s)$, which implies $h(s)=0$ and consequently $q(s)=0$. Therefore, the $2n$ polynomials $a(s),\ldots,s^{n-1}a(s),\widetilde d(s),\ldots,s^{n-1}\widetilde d(s)$ are linearly independent, and~\eqref{eq:simo_span_dimension} follows.
\end{proof}

\begin{remark}
Lemma~\ref{lem:simo_observable_output_combination} does not imply that one of the original pairs $(C_i,A)$ must be observable. The observable scalar output $\widetilde y=v\T y$ may be a nontrivial linear combination of several measured outputs. Therefore, in Case II, although $\zeta_{\widetilde y}(t)$ provides a convenient device for proving the total intrinsic dimension, the independent components selected directly from the original filtered vector may still be distributed among multiple output-filtered vectors, as characterized in Proposition~\ref{prop:simo_case_ii_dependence_structure}. The subsequent QR procedure selects these components from the original filtered vector $z(t)$.
\end{remark}

This result completes the SIMO dimension analysis. The MIMO result below is established separately through polynomial row-vector subspaces.

\subsection{MIMO Systems}
We now consider the general MIMO case with $m$ inputs and $p$ outputs. In contrast to the SIMO case, each output channel is driven by all input channels. Therefore, the scalar numerator $d_j(s)$ becomes a polynomial row vector. Define $D(s)=C\operatorname{adj}(sI-A)B\in\R^{p\times m}[s]$, and let $D_j(s)$ be the $j$th row of $D(s)$, namely
\begin{equation}
    D_j(s)=C_j\operatorname{adj}(sI-A)B.
\end{equation}
Thus, 
\begin{equation}
    Y_j(s)=\frac{D_j(s)}{a(s)}U(s),
    \qquad
    U(s)=\begin{bmatrix}U_1(s)&\cdots&U_m(s)\end{bmatrix}\T,
\end{equation}
and the filtered input--output variables correspond to the polynomial row-vector set
\begin{equation}
    \mathcal S_{\rm MIMO}=\{s^ra(s)e_i:r=0,\ldots,n-1,\ i=1,\ldots,m\}
    \cup\{s^rD_j(s):r=0,\ldots,n-1,\ j=1,\ldots,p\},
    \label{eq:mimo_polynomial_row_set}
\end{equation}
where $e_i$ is the $i$th canonical row vector in $\R^{1\times m}$. Thus, determining the number of independent components of $z(t)$ is equivalent to determining $\dim\spanop\mathcal S_{\rm MIMO}$.

\begin{lemma}
Let $\mathcal U=\spanop\{s^ra(s)e_i:r=0,\ldots,n-1,\ i=1,\ldots,m\}$. We have
\begin{equation}
    \dim\mathcal U=nm.
    \label{eq:input_generated_subspace_dimension}
\end{equation}
\label{lem:mimo_input_generated_subspace}
\end{lemma}

\begin{proof}
It suffices to show that the $nm$ polynomial row vectors $s^ra(s)e_i$, $i=1,\ldots,m$, $r=0,\ldots,n-1$, are linearly independent. Suppose that
\begin{equation}
    \sum_{i=1}^{m}\sum_{r=0}^{n-1}c_{ir}s^ra(s)e_i=0.
    \label{eq:input_generated_subspace_linear_combination}
\end{equation}
Equivalently,
\begin{equation}
    a(s)\begin{bmatrix}
    \sum_{r=0}^{n-1}c_{1r}s^r &
    \cdots &
    \sum_{r=0}^{n-1}c_{mr}s^r
    \end{bmatrix}=0.
    \label{eq:input_generated_subspace_component_form}
\end{equation}
Since $a(s)$ is a nonzero polynomial, each scalar polynomial in the bracket must be zero. Therefore, $c_{ir}=0$ for all $i=1,\ldots,m$ and $r=0,\ldots,n-1$. Hence the $nm$ generators of $\mathcal U$ are linearly independent, and thus $\dim\mathcal U=nm$. This completes the proof.
\end{proof}

\begin{lemma}
Let $M(s)=\operatorname{adj}(sI-A)B$, for each output channel $j=1,\ldots,p$ and $r=0,\ldots,n-1$, the polynomial row vector $s^rD_j(s)$ satisfies
\begin{equation}
    s^rD_j(s)\equiv C_jA^rM(s)\pmod{\mathcal U}.
    \label{eq:mimo_output_mod_input_subspace}
\end{equation}
\label{lem:mimo_output_mod_input_subspace}
\end{lemma}

\begin{proof}
Recall that $D_j(s)=C_jM(s)$, from
\begin{equation}
    (sI-A)\operatorname{adj}(sI-A)=a(s)I,
    \label{eq:adjugate_characteristic_identity}
\end{equation}
we obtain
\begin{equation}
    s\operatorname{adj}(sI-A)=A\operatorname{adj}(sI-A)+a(s)I.
    \label{eq:adjugate_shift_identity}
\end{equation}
Multiplying \eqref{eq:adjugate_shift_identity} by $B$ gives
\begin{equation}
    sM(s)=AM(s)+a(s)B.
    \label{eq:mimo_m_shift_identity}
\end{equation}
Applying \eqref{eq:mimo_m_shift_identity} repeatedly yields, for $r=1,\ldots,n-1$,
\begin{equation}
    s^rD_j(s)=C_jA^rM(s)
    +a(s)C_j\big(s^{r-1}I+s^{r-2}A+\cdots+A^{r-1}\big)B.
    \label{eq:mimo_output_repeated_identity}
\end{equation}
For $r=0$, this identity reduces to $D_j(s)=C_jM(s)$. The second term in \eqref{eq:mimo_output_repeated_identity} is of the form $a(s)Q_{jr}(s)$, where $Q_{jr}(s)\in\R^{1\times m}[s]$ has degree no larger than $r-1\leq n-2$. Hence it belongs to $\mathcal U$, because it can be represented by the generators $s^\ell a(s)e_i$, $\ell=0,\ldots,n-1$, $i=1,\ldots,m$. Therefore,
\begin{equation}
    s^rD_j(s)\equiv C_jA^rM(s)\pmod{\mathcal U}.
\end{equation}
This completes the proof.
\end{proof}

\begin{lemma}
Define
\begin{equation}
    \mathcal M=\mathbb R^{1\times n}M(s):=\{v\T M(s): v\in \mathbb R^{n}\}.
    \label{eq:output_quotient}
\end{equation}
Under Assumption~\ref{ass:controllable}, we have
\begin{equation}
    \dim\mathcal M=n.
    \label{eq:output_quotient_dimension}
\end{equation}
\label{lem:mimo_output_quotient}
\end{lemma}

\begin{proof}
The key is to show that the mapping $v\mapsto v\T M(s)$ is injective. If $v\T M(s)=0$, we can obtain
\begin{equation}
    v\T\operatorname{adj}(sI-A)B=0,
\end{equation}
which implies
\begin{equation}
    v\T(sI-A)^{-1}B=0.
    \label{eq:resolvent_zero_condition}
\end{equation}
Using the expansion $(sI-A)^{-1}=\frac{1}{s}I+\frac{1}{s^2}A+\frac{1}{s^3}A^2+\cdots$ gives
\begin{equation}
    v\T B=0,
    \quad v\T AB=0,
    \quad \ldots,
    \quad v\T A^{n-1}B=0.
\end{equation}
By Assumption~\ref{ass:controllable}, this gives $v=0$. Hence $\dim\mathcal M=n$.
\end{proof}

\begin{remark}
By Assumption~\ref{ass:observable}, the observability matrix $\mathcal O=[C\T\ (CA)\T\ \cdots\ (CA^{n-1})\T]\T$ has full column rank $n$. Therefore, the row vectors $C_jA^r$, $j=1,\ldots,p$, $r=0,\ldots,n-1$, span the whole row space $\mathbb R^{1\times n}$. It follows that $\spanop\{C_jA^rM(s):j=1,\ldots,p,\ r=0,\ldots,n-1\}=\mathcal M$. Thus, after removing the input-generated subspace $\mathcal U$, the output-filtered polynomial vectors in Lemma~\ref{lem:mimo_output_mod_input_subspace} generate exactly the quotient directions described by $\mathcal M$.
\end{remark}

\begin{lemma}
The input-generated subspace $\mathcal U$ and the output quotient space $\mathcal M$ satisfy
\begin{equation}
    \mathcal U\cap\mathcal M=\{0\}.
    \label{eq:input_output_trivial_intersection}
\end{equation}
\label{lem:mimo_input_output_independence}
\end{lemma}

\begin{proof}
Suppose that a polynomial row vector belongs to both spaces. Then there exist $Q(s)$ and $v$ such that
\begin{equation}
    a(s)Q(s)=v\T M(s).
    \label{eq:input_m_intersection_identity}
\end{equation}
The left-hand side has degree at least $n$ if $Q(s)$ is nonzero, while the right-hand side has degree at most $n-1$. Hence $Q(s)=0$ and $v=0$. Therefore, the intersection is trivial.
\end{proof}

\begin{theorem}
For a MIMO system satisfying Assumptions~\ref{ass:controllable} and~\ref{ass:observable}, with state dimension $n$, input dimension $m$, and output dimension $p$, the complete filtered input--output vector
\begin{equation}
    z(t)\in\R^{n(m+p)}
\end{equation}
contains only $n(m+1)$ linearly independent components. Equivalently,
\begin{equation}
    \dim\spanop\mathcal S_{\rm MIMO}=n(m+1).
    \label{eq:mimo_span_dimension}
\end{equation}
\label{thm:mimo_filtered_dimension}
\end{theorem}

\begin{proof}
By Lemma~\ref{lem:mimo_input_generated_subspace}, the input-generated subspace has dimension $nm$. By Lemma~\ref{lem:mimo_output_mod_input_subspace}, the output-filtered polynomial vectors, after removing the input-generated subspace, are represented by elements of $\mathcal M$. By Lemma~\ref{lem:mimo_output_quotient}, this output quotient space has dimension $n$. By Lemma~\ref{lem:mimo_input_output_independence}, the input-generated subspace and the output quotient space are independent. Hence, their dimensions add to $nm+n=n(m+1)$, which proves the theorem.
\end{proof}

Theorem~\ref{thm:mimo_filtered_dimension} gives the exact intrinsic dimension of the conventional MIMO filtered representation. Although $z(t)$ lies in the ambient space $\R^{n(m+p)}$, its components span a subspace of dimension only $n(m+1)$. Thus, when $p>1$, the complete parametrization contains $n(p-1)$ redundant directions. Using the full $z(t)$ in the Bellman regression therefore introduces deterministically dependent regressors and may cause rank deficiency.

\begin{corollary}
Under the conditions of Theorem~\ref{thm:mimo_filtered_dimension}, there exists a selection matrix $S$ such that the $n(m+1)$ components of
\begin{equation}
    z_r(t)=Sz(t)\in\R^{n(m+1)}
    \label{eq:selection_reduced_filtered_vector}
\end{equation}
are linearly independent. Moreover, the state parametrization can be preserved in the reduced form
\begin{equation}
    x(t)\approx M_rz_r(t)
    \label{eq:reduced_state_parametrization_corollary}
\end{equation}
for some constant matrix $M_r$. 
\label{cor:reduced_filtered_vector}
\end{corollary}

The above result provides the theoretical basis for the reduced model-free output feedback formulation. In the next section, we show how to construct $z_r(t)$ directly from data and how to reformulate the PI and VI Bellman equations using $z_r(t)$ instead of the full vector $z(t)$, whose components are linearly dependent.

\section{Reduced Model-Free Output Feedback LQR}
Corollary~\ref{cor:reduced_filtered_vector} translates the intrinsic-dimension characterization into a reduced state parametrization for model-free LQR learning. In the original dynamic output feedback formulation in~\cite{Rizvi2020}, one directly uses $z(t)\in\R^{n(m+p)}$ to construct the Bellman equations. However, since the intrinsic dimension of $z(t)$ is only $n(m+1)$, the regression matrix generated by
$
    \vecs\big(z(t)z\T(t)\big), u(t)\otimes z(t)
$
in~\eqref{eq:full_filtered_pi_regression} and~\eqref{eq:full_filtered_vi_regression} is rank deficient. Therefore, even if the input--output data are sufficiently excited, the linear equations for solving the unknown matrices $\bar P_k$, $\bar K_k$, or $\bar H_k$ in~\eqref{eq:full_filtered_pi_regression} and~\eqref{eq:full_filtered_vi_regression} fail to have a unique solution due to the deterministic linear dependence inside $z(t)$.

To avoid this problem, the model-free LQR Bellman equations should be written in terms of $z_r(t)$, rather than the original full vector $z(t)$. By \eqref{eq:reduced_state_parametrization_corollary}, the value function and control law can be represented as
\begin{equation}
    V(t)=z_r\T(t)P_rz_r(t), \quad u(t)=K_rz_r(t),
\end{equation}
where $P_r\in \mathbb{R}^{n(m+1)\times n(m+1)}$ and $K_r \in\mathbb{R}^{m\times n(m+1)}$.

In practice, the independent components of $z(t)$ can be selected directly from data. Suppose the filtered data are collected at sampling instants $t_1,t_2,\ldots,t_N\ (N\ge n(m+1))$. Define $Z=\begin{bmatrix}z(t_1)&z(t_2)&\cdots&z(t_N)\end{bmatrix}\T\in\R^{N\times n(m+p)}$ and perform a QR decomposition with column pivoting,
let $r=n(m+1)$. Then the first $r$ pivot columns selected by the QR decomposition correspond to a linearly independent subset of the components of $z(t)$. With the selected index set $\mathcal I=\{i_1,i_2,\ldots,i_r\}$, the reduced filtered vector can be defined as
\begin{equation}
    z_r(t)=\begin{bmatrix}
    z_{i_1}(t)&z_{i_2}(t)&\cdots&z_{i_r}(t)
    \end{bmatrix}\T\in\R^r.
    \label{eq:qr_reduced_filtered_vector}
\end{equation}
\begin{remark}
    The step of QR decomposition is used only to determine which components of the linearly dependent filtered vector $z(t)$ should be retained. The selected index set $\mathcal I$ is fixed and used in all subsequent learning iterations.
\end{remark}

\subsection{Reduced PI}

Substituting the reduced state parametrization~\eqref{eq:reduced_state_parametrization_corollary} into the PI equation~\eqref{eq:model_free_pi_integral} gives the reduced PI equation
\begin{align}
&z_r\T(t)P_{r,k}z_r(t)
-z_r\T(t-T)P_{r,k}z_r(t-T)\nonumber\\
&= -\int_{t-T}^{t}y\T(\tau)Q_y y(\tau)d\tau
   -\int_{t-T}^{t}z_r\T(\tau)K_{r,k}\T R K_{r,k}z_r(\tau)d\tau\nonumber\\
&\quad -2\int_{t-T}^{t}
\left(u(\tau)-K_{r,k}z_r(\tau)\right)\T
RK_{r,k+1}z_r(\tau)d\tau.
\label{eq:reduced_pi_equation}
\end{align}
Stacking \eqref{eq:reduced_pi_equation} over all data intervals gives the reduced PI regression
\begin{equation}
\Psi_{r,k}
\begin{bmatrix}
\operatorname{vecs}(P_{r,k})\\
\operatorname{vec}(K_{r,k+1})
\end{bmatrix}
=
\Phi_{r,k},
\label{eq:reduced_pi_regression}
\end{equation}
where $\Psi_{r,k}=\big[\delta_{z_rz_r},-2\Gamma_{z_rz_r}(I_r\otimes K_{r,k}\T R)+2\Gamma_{z_ru}\big]$ and $\Phi_{r,k}=-I_{yy}\operatorname{vecs}(Q_y)-\Gamma_{z_rz_r}\operatorname{vecs}(K_{r,k}\T RK_{r,k})$.

Here, the number of unknowns is
$\frac{r(r+1)}{2}+mr$, where $r=n(m+1)$.
Therefore, compared with the original output feedback formulation using $z\in\R^{n(m+p)}$, the reduced formulation decreases the number of unknowns from
$\frac{n(m+p)\left(n(m+p)+1\right)}{2}+mn(m+p)$
to
$\frac{n(m+1)\left(n(m+1)+1\right)}{2}+mn(m+1)$.
More importantly, the reduced formulation removes the linear dependence in $z(t)$, thereby making the least-squares equation well posed under sufficient excitation. The complete reduced PI procedure is summarized in Algorithm~\ref{alg:reduced_pi}, where the QR-based reduction step implements the data-driven construction of $z_r(t)$ and the policy evaluation/improvement step is based on \eqref{eq:reduced_pi_regression}.

During data collection, the exploratory input is chosen as:
\begin{equation}
    u_0(t)=K_{r,0}z_r(t)+\nu(t),
    \label{eq:behavior_input_pi}
\end{equation}
where $K_{r,0}$ is stabilizing and $\nu(t)$ is the exploration signal. The resulting input--output data are used to generate $z(t)$ and are reused in the subsequent policy evaluation and improvement steps.

\begin{remark}
    It should be noted that PI requires an admissible initial policy. In the existing model-free LQR PI literature~\cite{Rizvi2020,Chen2025Homotopy,Lin2026}, this issue is handled either by considering open-loop stable systems or by introducing a homotopy-based initialization. For an open-loop stable system, the zero feedback gain is admissible, and hence the initial input becomes $u_0(t)=\nu(t)$.
\end{remark}

\begin{remark}
The above reduction procedure is completely data-based. It only uses the sampled filtered data matrix $Z$ and the data matrices generated from input--output data. Therefore, it does not require the knowledge of the system matrices $A$, $B$, or $C$.
\end{remark}

\begin{remark}
In implementation, the reduced data matrices such as $\delta_{z_rz_r}$, $\Gamma_{z_rz_r}$, and $\Gamma_{z_ru}$ in~\eqref{eq:reduced_pi_regression} are not directly reconstructed from the reduced vector $z_r(t)$ in~\eqref{eq:qr_reduced_filtered_vector}. Instead, the full matrices $\delta_{zz}$, $\Gamma_{zz}$, $\Gamma_{zu}$, $I_{zz}$, and $I_{zu}$ are first constructed using the original full vector $z(t)$. Then the reduced matrices are obtained by extracting the columns corresponding to the selected indices in $\mathcal I$. This avoids additional numerical integration errors caused by recomputing the integral terms from the sampled reduced data.
\end{remark}

\begin{algorithm}[!t]
\caption{Reduced Model-Free LQR PI Using Output Feedback}
\label{alg:reduced_pi}
\begin{algorithmic}[1]
\State \textbf{Initialize and Data Collection:} Select a stabilizing policy $u_0(t)$ in~\eqref{eq:behavior_input_pi}, and set $k\gets0$. Apply $u_0(t)$ and generate $z(t)$, $\delta_{zz}$, $\Gamma_{zz}$, $\Gamma_{zu}$, $I_{zz}$, and $I_{zu}$ from~\eqref{eq:z_system} on $[t_0,t_l]$, where $t_l=t_0+lT$. 
\State \textbf{QR-Based Reduction:} Perform a QR decomposition with column pivoting on $Z$, generate $z_r(t)$ from the selected index set $\mathcal I$, and fix this reduced vector for all learning iterations.
\State \textbf{Data Matrix Reconstruction:} Reconstruct the reduced data matrices, such as $\delta_{z_rz_r}$, $\Gamma_{z_rz_r}$, $\Gamma_{z_ru}$, $I_{z_rz_r}$, and $I_{z_ru}$, from the full data matrices according to the index set $\mathcal I$. 
\Loop
    \State \textbf{Evaluate and Improve Policy:} Solve \eqref{eq:reduced_pi_regression} for $P_{r,k}$ and $K_{r,k+1}$ by least squares.
    \If{$\|P_{r,k+1}-P_{r,k}\|<\varepsilon$}
        \State \Return $K_r=K_{r,k}$.
    \EndIf
    \State $k\gets k+1$.
\EndLoop
\end{algorithmic}
\end{algorithm}

\subsection{Reduced VI}

PI requires an initial stabilizing policy. To remove this requirement, the same reduced formulation can be used to derive a reduced VI method. Substituting the reduced state parametrization~\eqref{eq:reduced_state_parametrization_corollary} into the VI equation~\eqref{eq:model_free_vi_integral} gives the reduced VI equation
\begin{align}
&z_r\T(t)P_{r,k}z_r(t)
-z_r\T(t-T)P_{r,k}z_r(t-T)\nonumber\\
&=-\int_{t-T}^{t}y\T(\tau)Q_y y(\tau)d\tau
+\int_{t-T}^{t}z_r\T(\tau)H_{r,k}z_r(\tau)d\tau\nonumber\\
&\quad -2\int_{t-T}^{t}(Ru(\tau))\T K_{r,k}z_r(\tau)d\tau.
\label{eq:reduced_vi_equation}
\end{align}
Stacking \eqref{eq:reduced_vi_equation} over all data intervals gives the reduced VI regression
\begin{equation}
\Omega_r
\begin{bmatrix}
\operatorname{vecs}(H_{r,k})\\
\operatorname{vec}(K_{r,k})
\end{bmatrix}
=
\Theta_{r,k},
\label{eq:reduced_vi_regression}
\end{equation}
where $\Omega_r=\big[I_{z_rz_r}\ -2I_{z_ru}\big]$ and $\Theta_{r,k}=\delta_{z_rz_r}\operatorname{vecs}(P_{r,k})+I_{yy}\operatorname{vecs}(Q_y)$.
Here, the unknowns in this regression are $H_{r,k}$ and $K_{r,k}$. Once they are obtained by least squares, the tentative value matrix is computed by
\begin{equation}
    \widetilde P_{r,k+1}=P_{r,k}+\epsilon_k(H_{r,k}-K_{r,k}\T RK_{r,k}),
\end{equation}
where the stepsize sequence $\{\epsilon_k\}_{k=0}^\infty$ satisfies the same stepsize conditions as in~\eqref{eq:model_based_vi_conditions}. For the reduced formulation, the same bounded-set construction is used in $\mathbb{P}^{n(m+1)}_+$:
\begin{equation}
    \mathcal{B}_q\subseteq \mathcal{B}_{q+1},\quad q\in\mathbb{N},\qquad
    \lim_{q\to\infty}\mathcal{B}_q=\mathbb{P}^{n(m+1)}_+.
\end{equation}
The tentative matrix is accepted or reset using the same truncation mechanism described above. Following the VI strategy in~\cite{Rizvi2020}, the complete reduced VI procedure is given in Algorithm~\ref{alg:reduced_vi}.

\begin{algorithm}[!t]
\caption{Reduced Model-Free LQR VI Using Output Feedback}
\label{alg:reduced_vi}
\begin{algorithmic}[1]
\State \textbf{Initialize and Data Collection:} Select a control policy $u_0(t)$ in~\eqref{eq:behavior_input_pi} and $P_{r,0}>0$, and set $k\gets0$ and $q\gets0$.
Apply $u_0(t)$ and generate $z(t)$, $\delta_{zz}$, $\Gamma_{zz}$, $\Gamma_{zu}$, $I_{zz}$, and $I_{zu}$ from~\eqref{eq:z_system} on $[t_0,t_l]$, where $t_l=t_0+lT$. 
\State \textbf{QR-Based Reduction:} Perform a QR decomposition with column pivoting on $Z$, generate $z_r(t)$ from the selected index set $\mathcal I$, and fix this reduced vector for all learning iterations.
\State \textbf{Data Matrix Reconstruction:} Reconstruct the reduced data matrices, such as $\delta_{z_rz_r}$, $\Gamma_{z_rz_r}$, $\Gamma_{z_ru}$, $I_{z_rz_r}$, and $I_{z_ru}$, from the full data matrices according to the index set $\mathcal I$. 
\Loop
    \State \textbf{Loop:} Solve \eqref{eq:reduced_vi_regression} for $H_{r,k}$ and $K_{r,k}$ by least squares.
    \State \textbf{Update Value Matrix:} Compute $\widetilde P_{r,k+1}=P_{r,k}+\epsilon_k\left(H_{r,k}-K_{r,k}\T RK_{r,k}\right)$.
    \If{$\widetilde P_{r,k+1}\notin \mathcal{B}_q$}
        \State $P_{r,k+1}\gets P_{r,0}$ and $q\gets q+1$.
    \ElsIf{$\|\widetilde P_{r,k+1}-P_{r,k}\|/\epsilon_k<\varepsilon$}
        \State \Return $K_r=K_{r,k}$.
    \Else
        \State $P_{r,k+1}\gets\widetilde P_{r,k+1}$.
    \EndIf
    \State $k\gets k+1$.
\EndLoop
\end{algorithmic}
\end{algorithm}
Compared with reduced PI, reduced VI does not require an initially stabilizing policy. Therefore, for the exploratory input in~\eqref{eq:behavior_input_pi}, the initial feedback gain $K_{r,0}$ can be set to zero, so that one usually simply choose $u_0(t)=\nu(t)$, as in the PI case. However, compared with PI, it usually needs more iterations to converge.

\begin{remark}
It should be emphasized that the linear independence of $z_r(t)$ is a requirement, while the persistent excitation of the collected data is still needed to guarantee that the regression matrix has full column rank. The QR decomposition removes deterministic linear dependence among the filtered variables, and the exploration signal $\nu(t)$ provides sufficient excitation for solving the least-squares problem.
\end{remark}

\section{Numerical Example}

Consider the jet transport aircraft model, a classical $2\times2$ MIMO benchmark system taken from~\cite{MathWorksMIMOJet}, in the form of \eqref{eq:plant_system}, with the matrices $A$, $B$, and $C$ given as follows:

\[
\begin{aligned}
A &=
\begin{bmatrix}
-0.0558 & -0.9968 &  0.0802 & 0.0415 \\
 0.5980 & -0.1150 & -0.0318 & 0      \\
-3.0500 &  0.3880 & -0.4650 & 0      \\
 0      &  0.0805 &  1.0000 & 0
\end{bmatrix}, \\
B &=
\begin{bmatrix}
 0.0073 & 0      \\
-0.4750 & 0.0077 \\
 0.1530 & 0.1430 \\
 0      & 0
\end{bmatrix}, \qquad
C =
\begin{bmatrix}
0 & 1 & 0 & 0 \\
0 & 0 & 0 & 1
\end{bmatrix}.
\end{aligned}
\]
Here, $n=4$, $m=2$, and $p=2$. It can be seen that Assumptions~\ref{ass:controllable} and~\ref{ass:observable} hold for this example. The LQR weighting matrices are selected as $Q_y=I_2$ and $R=I_2$. The Hurwitz filter polynomial is chosen as $\Lambda(s)=(s+2)^4$.

To implement the proposed reduced model-free output feedback learning algorithm, the data are collected under the exploratory input $u_0(t)=K_{r,0}z_r(t)+\nu(t)$ with $K_{r,0}=\textbf{0}$, where $\nu(t)$ is a combination of sinusoidal signals with different frequencies as in~\cite{Jiang2012}.
Let $t_j=j \Delta t $, $j=0,1,\ldots,N$, where $\Delta t=0.02$~s and $N=500$. The data are collected over $t\in[0,10]$~s to construct the filtered input--output vector $z(t)$.

For this example, the original filtered vector satisfies $z(t)\in\mathbb R^{n(m+p)}=\mathbb R^{16}$.
However, after performing a column-pivoted QR decomposition on the filtered data matrix $Z$ defined above, the numerical rank is obtained as $\rank(Z)=12$, which agrees with the theoretical result. The first 12 pivot columns are retained to form the reduced filtered vector $z_r(t)\in\mathbb R^{12}$. The proposed reduced model-free output feedback VI algorithm is then implemented using $z_r(t)$.

To quantify the effect of the reduction on the Bellman regression, the full matrix $\bar\Omega$ in~\eqref{eq:full_filtered_vi_regression} and the reduced matrix $\Omega_r$ in~\eqref{eq:reduced_vi_regression} are constructed from the same data. Table~\ref{tab:regression_comparison} shows that $\bar\Omega$ has 168 columns but rank 102, whereas $\Omega_r$ has full column rank. Moreover, the condition number decreases from $3.6\times10^{17}$ to $2.4\times10^6$, indicating a substantial improvement in numerical conditioning.

\begin{table}[t]
\centering
\caption{Comparison of the VI Regression Matrices in~\eqref{eq:full_filtered_vi_regression} and~\eqref{eq:reduced_vi_regression}}
\label{tab:regression_comparison}
\scriptsize
\renewcommand{\arraystretch}{1.1}
\begin{tabular}{@{}lcccc@{}}
\toprule
Regression matrix & Unknowns & Rank & Columns & Condition no. \\
\midrule
$\bar\Omega$ in~\eqref{eq:full_filtered_vi_regression}
& 168 & 102 & 168 & $3.6\times10^{17}$ \\
$\Omega_r$ in~\eqref{eq:reduced_vi_regression}
& 102 & 102 & 102 & $2.4\times10^4$ \\
\bottomrule
\end{tabular}
\end{table}

The control parameters are initialized to zero. Step size is selected as $\epsilon_k=(k+5)^{-1}, k\in\mathbb{N}$ and the bounded set $\mathcal{B}_q=\{P\in \mathbb{P}_+^{12}:\|P\|\le 10000(q+1)\}, q\in \mathbb{N}$. Thus, under the convergence criterion, the learned reduced output feedback gain is obtained as follows, together with the model-based optimal reduced gain $K_r^*$ for comparison:
\begingroup
\scriptsize
\setlength{\arraycolsep}{2pt}
\renewcommand{\arraystretch}{0.92}
\[
\begin{gathered}
K_r=\begin{bmatrix}K_r^{1}&K_r^{2}\end{bmatrix},\qquad
K_r^*=\begin{bmatrix}K_r^{*,1}&K_r^{*,2}\end{bmatrix},\\[0.5mm]
K_r^{1}=\begin{bmatrix}
-84.7145 & 12.3461\\
-61.6217 &  7.8132\\
-16.0778 &  1.6581\\
 -1.8159 &  0.1708\\
  4.0055 & -0.5315\\
 -2.8204 &  0.4777
\end{bmatrix}^{\top},\quad
K_r^{2}=\begin{bmatrix}
  1.3874 & -0.3136\\
  0.1708 & -0.0373\\
 15.2226 & -4.0348\\
 28.3082 & -9.2522\\
 33.6469 &-11.0323\\
 49.0950 &-10.3153
\end{bmatrix}^{\top},\\[0.8mm]
K_r^{*,1}=\begin{bmatrix}
-84.7054 & 12.3622\\
-61.6154 &  7.8237\\
-16.0763 &  1.6604\\
 -1.8157 &  0.1710\\
  4.0117 & -0.5164\\
 -2.8144 &  0.4914
\end{bmatrix}^{\top},\quad
K_r^{*,2}=\begin{bmatrix}
  1.3892 & -0.3096\\
  0.1710 & -0.0368\\
 15.2165 & -4.0429\\
 28.3073 & -9.2572\\
 33.6444 &-11.0460\\
 49.0895 &-10.3290
\end{bmatrix}^{\top}.
\end{gathered}
\]
\endgroup
The relative Frobenius error between $K_r$ and $K_r^*$ is $3.08\times10^{-4}$, indicating close agreement with the model-based optimum. Consistent with the regression-matrix comparison in Table~\ref{tab:regression_comparison}, the full formulation exhibits repeated resets and fails to converge in Fig.~\ref{fig:full_numerical_convergence}, whereas the reduced formulation converges in Fig.~\ref{fig:numerical_convergence}. After learning, the exploration signal is removed, and the learned controller $u(t)=K_rz_r(t)$ is applied to the system. All state trajectories converge to zero in Fig.~\ref{fig:numerical_closed_loop}. Taken together, the improved regression rank and conditioning, the accurate learned gain, the convergence behavior, and the closed-loop response demonstrate the effectiveness of the proposed reduction method.

\begin{figure}[t]
\centering
\includegraphics[width=\columnwidth]{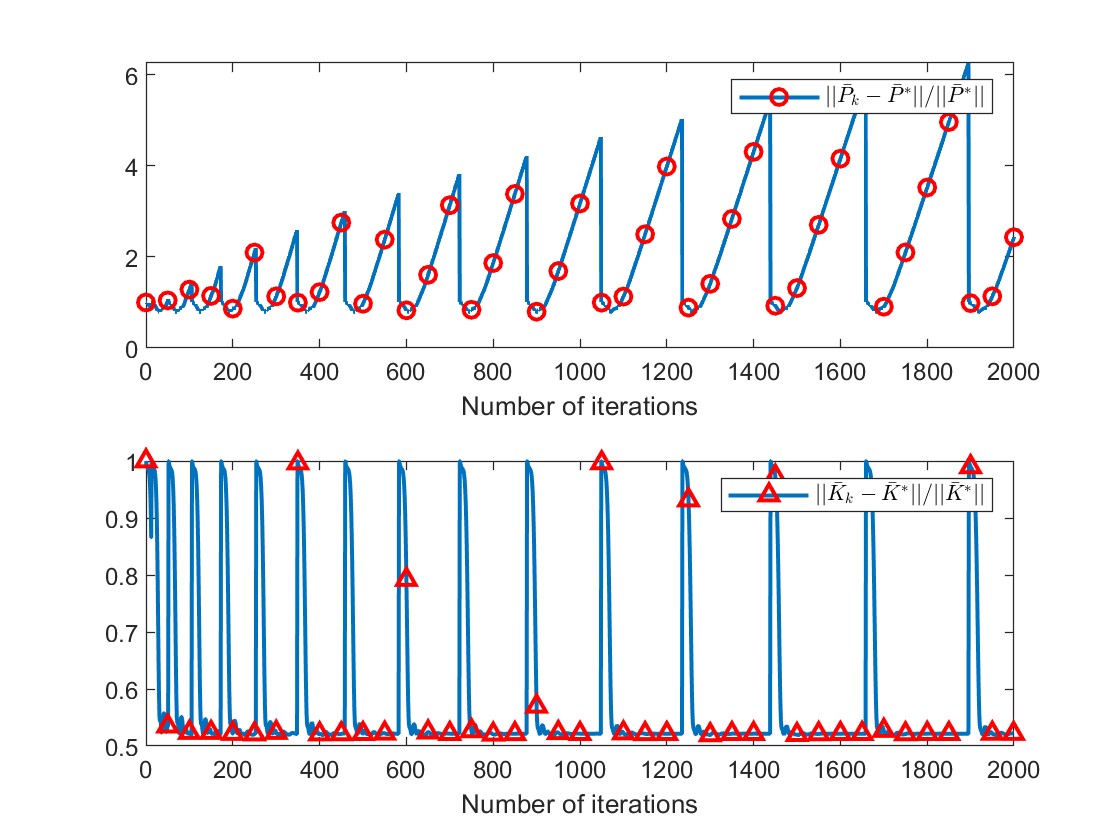}
\caption{Iteration curves of the full VI iterates $\bar P_k$ and $\bar K_k$ versus the iteration number.}
\label{fig:full_numerical_convergence}
\end{figure}

\begin{figure}[t]
\centering
\includegraphics[width=\columnwidth]{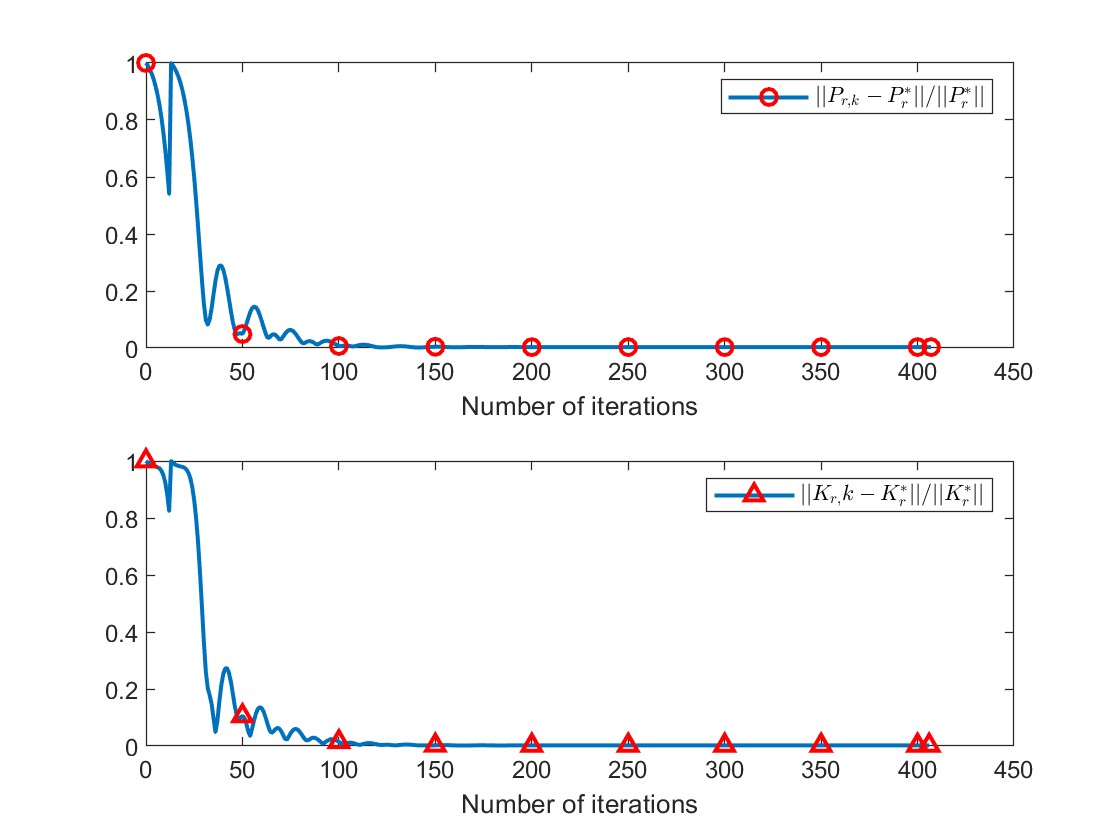}
\caption{Convergence curves of the reduced VI iterates $P_{r,k}$ and $K_{r,k}$ versus the iteration number.}
\label{fig:numerical_convergence}
\end{figure}

\begin{figure}[t]
\centering
\includegraphics[width=\columnwidth]{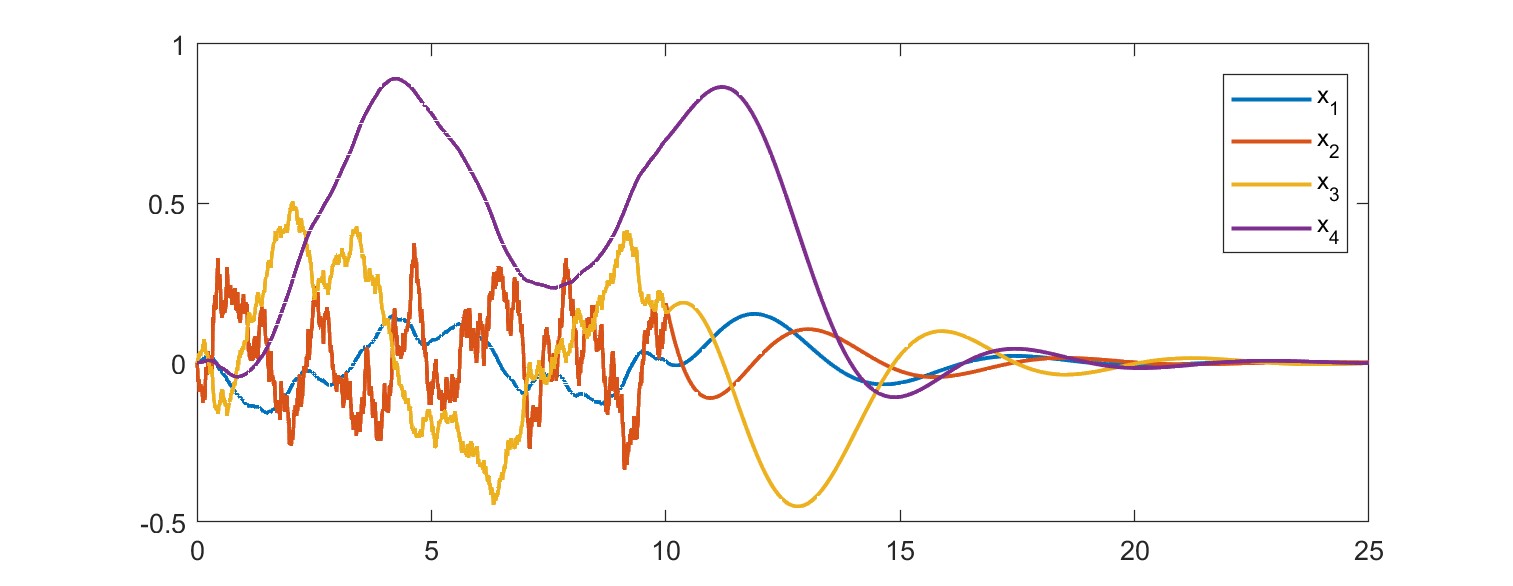}
\caption{Closed-loop response under the learned reduced output feedback controller: state trajectories.}
\label{fig:numerical_closed_loop}
\end{figure}


\section{Conclusion}

This article characterized the intrinsic dimension of the filter-based state parametrization used in model-free output feedback LQR learning. The linear-dependence analysis showed that the conventional filtered vector contains only $2n$ independent components for SIMO systems and $n(m+1)$ independent components for general MIMO systems, despite its larger ambient dimension. Based on this characterization, a QR-based reduction method was introduced to construct an independent reduced vector, and the model-free PI/VI equations were reformulated accordingly. The numerical example verified that the reduced formulation removes the rank deficiency of the Bellman regression equation and yields an effective output feedback LQR controller. Future work will consider the robustness of the proposed method under noisy and finite data, as well as its extension to online implementation and more general nonlinear systems.

\appendix
\section{Proof of Lemma~\ref{lem:simo_case_i_representation}}
\label{app:simo_case_i_lemma_proof}
Since 
\begin{equation}
    \gcd(a(s),d_i(s))=1,
\end{equation}
B\'ezout's identity implies that, for any polynomial $g(s)$, there exist polynomials $\bar\alpha(s)$ and $\bar\beta(s)$ satisfying
\begin{equation}
    \bar\alpha(s)a(s)+\bar\beta(s)d_i(s)=g(s).
    \label{eq:appendix_a_bezout_identity}
\end{equation}
In Lemma~\ref{lem:simo_case_i_representation}, let $g(s)=s^rd_j(s), r=0,\ldots,n-1$.
Since $\deg d_j(s)\leq n-1$, we have $\deg g(s)\leq 2n-2$. Starting from one B\'ezout representation~\eqref{eq:appendix_a_bezout_identity},
divide $\bar\beta(s)$ by $a(s)$, there exist polynomials $q(s)$ and $\beta(s)$ such that
\begin{equation}
    \bar\beta(s)=q(s)a(s)+\beta(s),
    \qquad
    \deg\beta(s)<n.
    \label{eq:appendix_a_beta_division}
\end{equation}
Substituting \eqref{eq:appendix_a_beta_division} into \eqref{eq:appendix_a_bezout_identity} gives
\begin{equation}
    \tilde\alpha(s)a(s)+\beta(s)d_i(s)=g(s),
    \qquad
    \deg\beta(s)<n.
    \label{eq:appendix_a_reduced_beta_representation}
\end{equation}
where
\[
    \tilde\alpha(s)=\bar\alpha(s)+q(s)d_i(s).
    \label{eq:appendix_a_alpha_tilde_definition}
\]

It remains to show that the coefficient $\tilde\alpha(s)$ in \eqref{eq:appendix_a_reduced_beta_representation} can also be chosen with degree less than $n$. From \eqref{eq:appendix_a_reduced_beta_representation},
\begin{equation}
    \tilde\alpha(s)a(s)=g(s)-\beta(s)d_i(s).
    \label{eq:appendix_a_alpha_tilde_identity}
\end{equation}
Since $\deg g(s)\leq 2n-2$, and $\deg\beta(s)<n, \deg d_i(s)\leq n-1$,
we have $\deg\big(g(s)-\beta(s)d_i(s)\big)\leq 2n-2$.
Because $\deg a(s)=n$, \eqref{eq:appendix_a_alpha_tilde_identity} implies that, if $\tilde\alpha(s)\neq 0$, we have
\begin{equation}
    \deg\tilde\alpha(s)\leq n-2<n.
    \label{eq:appendix_a_alpha_degree_bound}
\end{equation}
Thus, by taking
\begin{equation}
    \alpha_{jr}(s)=\tilde\alpha(s),
    \qquad
    \beta_{jr}(s)=\beta(s),
    \label{eq:appendix_a_coefficients_definition}
\end{equation}
we obtain
\begin{equation}
    \alpha_{jr}(s)a(s)+\beta_{jr}(s)d_i(s)=s^rd_j(s),
    \label{eq:appendix_a_final_representation}
\end{equation}
with $\deg\alpha_{jr}(s)<n, \deg\beta_{jr}(s)<n$. This completes the proof.

\end{document}